\documentclass[letter, 10 pt, conference]{ieeeconf}  

\IEEEoverridecommandlockouts                              

\usepackage{amsbsy}
\usepackage{floatflt} 
\usepackage{stackrel}
\usepackage{amsmath}
\usepackage{amssymb}
\usepackage{times}
\usepackage{graphicx}
\usepackage{xspace}
\usepackage{paralist} 
\usepackage{setspace} 
\usepackage{xypic}
\xyoption{curve}
\usepackage{latexsym}
\usepackage{theorem}
\usepackage{ifthen}
\usepackage{subfigure}
\usepackage[ruled,vlined]{algorithm2e}
\usepackage{booktabs}
\usepackage{algorithmic}

{\theoremheaderfont{\it} \theorembodyfont{\rmfamily}
\newtheorem{theorem}{Theorem}
\newtheorem{lemma}[theorem]{Lemma}

\newtheorem{corollary}[theorem]{Corollary}

}

\renewcommand{\baselinestretch}{1.}

{\theoremheaderfont{\it} \theorembodyfont{\rmfamily}

}

\begin{document}
\title{\Huge{Distributed Storage Allocations for Neighborhood-based Data Access}}


\author{Du$\breve{\mbox{s}}$an Jakoveti\'c, Aleksandar Minja, Dragana Bajovi\'c, and Dejan
Vukobratovi\'c
\thanks{D. Jakoveti\'c and D. Bajovi\'c are with University of Novi Sad, BioSense Center, Novi Sad, Serbia.
 A. Minja and D. Vukobratovi\'c
are with Department of Power, Electronics, and Communications Engineering,
University of Novi Sad, Novi Sad, Serbia. Authors' emails: {djakovet}@uns.ac.rs, sale.telekom@gmail.com, {[dbajovic,dejanv]}@uns.ac.rs.
}}

\maketitle

\begin{abstract}
We introduce a neighborhood-based data access model for distributed coded storage allocation. Storage nodes
 are connected in a generic network and data is accessed locally: a user
 accesses a randomly chosen storage node, which subsequently queries its neighborhood to recover the data object. We aim at finding an optimal allocation that minimizes the overall storage budget while ensuring recovery with probability one. We show that the problem reduces to finding the fractional dominating set of the underlying network.
   Furthermore, we develop a fully distributed algorithm where each storage node communicates only with its neighborhood in order to find its optimal storage allocation. The proposed algorithm is based upon the recently proposed proximal center method--an efficient dual decomposition based on accelerated dual gradient method.
 We show that our algorithm achieves a $(1+\epsilon)$-approximation ratio in $O(d_{\mathrm{max}}^{3/2}/\epsilon)$ iterations and per-node communications, where
 $d_{\mathrm{max}}$ is the maximal degree across nodes. Simulations demonstrate the effectiveness of the algorithm.
\end{abstract}

\vspace{-2mm}
\section{Introduction}
\label{section-intro}
With distributed (coded) storage allocation problems~\cite{Dimakis1,Dimakis2}, one aims to
 store a data object over $N$ storage nodes, such that the tradeoff between
 redundancy (total amount of storage) and reliability of accessing the object is balanced in an optimal way.
 A standard instance of the problem is as follows. Store a unit size data object~$D$ over $N$ nodes, such
  that each node $i$ stores an amount of $x_i$ storage encoded from object~$D$.
  At a later time, a data collector--user accesses a fixed number $r$, $r<N$,
  of randomly chosen nodes, and attempts to recover~$D$. Assuming a maximum distance separable~(MDS) coding is used, the recovery is successful if the overall amount of storage
  across the selected nodes is at least equal in size to~$D$. Then, the goal
  is to optimize the $x_i$'s such that the total storage
  $\sum_{i=1}^N x_i$ is minimized, while the recovery probability exceeds a prescribed level.

 In this paper--motivated by applications like cloud storage systems, peer-to-peer (P2P) networks, sensor networks or caching in small-cell cellular networks--we introduce a new, neighborhood-based data access model in the context of
 distributed coded storage allocation. We assume that the $N$ storage nodes
 are interconnected with links and constitute a generic network.
 A user (e.g., smartphone, sensor, P2P client application) accesses a randomly chosen node~$i$. Subsequently,
 node~$i$ contacts its one-hop neighbors, receives their coded storage,
 and, combined with its own storage, passes the aggregate storage to
 the user, which finally attempts the recovery.
 Then, our goal is to minimize $\sum_{i=1}^N x_i$ such that the recovery probability is one.

We show that the resulting problem is a fractional dominating set~(FDS) linear program~(LP), e.g.,~\cite{Kuhn}, and hence can be efficiently solved via standard LP solvers.
 However, we are interested in solving FDS
 in a \emph{fully distributed} way, whereby nodes over iterations exchange
 messages with their neighbors in the network with the aim of finding their optimal local allocations. Several
  distributed algorithms to solve FDS, and, more generally, fractional packing and covering LPs, have been proposed, ~\nocite{Luby,GlobalOptLocalInf,Stateless,MixedPackingCovering} e.g.,~\cite{Luby}-\cite{MixedPackingCovering}.
   However, \emph{advanced dual decomposition} techniques based on the Lagrangian dual, proved useful in many distributed applications, have not been sufficiently explored. In this paper, to solve FDS in a fully distributed way, we apply and modify the proximal center method proposed in~\cite{Necoara2}--an efficient dual decomposition based on accelerated Nesterov gradient method. We show that the resulting method is competitive with existing, primal-based approaches.
    Assuming that all nodes know $d_{\mathrm{max}}$ and $\epsilon$ beforehand ($d_{\mathrm{max}}$ the maximal degree across nodes, and $\epsilon$ the required accuracy), we show that the algorithm achieves a $(1+\epsilon)$-approximation ratio in~$O(d_{\mathrm{max}}^{3/2}/\epsilon)$ iterations~$k$--more precisely--$O(d_{\mathrm{max}}^{3/2}/\epsilon)$ per-node scalar communications
   and $O(d_{\mathrm{max}}^{5/2}/\epsilon)$ per-node elementary operations~(computational cost).
   This matches the best dependence on~$\epsilon$
   among existing \nocite{Luby,Bienstock,FasterSimpler,GlobalOptLocalInf,Stateless,MixedPackingCovering}
   solvers~\cite{Luby}-\cite{MixedPackingCovering},~\cite{GlobalOptLocalInf,FasterSimpler}, proposed for FDS (or fractional packing/covering).
   Furthermore, the algorithm's iterates are feasible (satisfy problem constraints) at all iterations.
   Simulations demonstrate that our algorithm converges much faster
   than a state-of-the-art distributed solver~\cite{Stateless} on moderate-size networks ($N=100-400$.)
    With respect to~\cite{Necoara2} (which considers generic convex programs), we introduce
here a novel, simple way of maintaining feasibility along iterations by exploiting the FDS problem structure.
 Further, exploiting structure, the results in~\cite{Necoara}, and primal-dual solution bounds
 derived here, we improve the dependence on the underlying network (on $N,d_{\mathrm{max}}$) with respect to a direct
 application of the generic results in~\cite{Necoara2}.

Summarizing, our main contributions are two-fold. First, we introduce a new, neighborhood-based data access model
for distributed coded storage allocation--the model well-suited in
many user-oriented applications for distributed networked storage--and we show that the corresponding optimization problem reduces to FDS.
Second, we solve the coded storage allocation problem (FDS) where storage nodes search for their optimal local allocations in a fully distributed way by applying and modifying the dual-based proximal center method. We establish the method's convergence and complexity guarantees and show that it compares favorably with a state-of-the art method on moderate-size networks.



\subsection{Literature review, paper organization, and notation}
We now briefly review the literature to further contrast our paper with
 existing work. Neighborhood-based data access has been
 previously considered in the context of replica placement, e.g.,~\cite{ReplicaPlacement}.
 Therein, one wants to replicate the \emph{raw object}~$D$
 across network such that it is reliably accessible through a neighborhood of any node.
  In contrast, we consider here the \emph{coded storage}.
   Mathematically, replica placement corresponds to the integer dominating set problem (known to be NP hard),
  while coded allocation studied here translates into FDS (solvable in polynomial time).

We now contrast
our distributed algorithm for FDS with existing methods.
%
  The literature usually considers more general fractional packing/covering
 problems, \nocite{Luby,Bienstock,FasterSimpler,GlobalOptLocalInf,Stateless,MixedPackingCovering}, e.g.,
  \cite{Luby}-\cite{MixedPackingCovering},~\cite{GlobalOptLocalInf,FasterSimpler}, and we hence specialize their results to FDS.\footnote{The constraint matrix $A$ is in our case square, $N \times N$, and it has $0/1$ entries,
  so that the width of the problem (largest entry of~$A$) is one.}
  References~\cite{Luby}-\cite{MixedPackingCovering}
  develop distributed algorithms, with the required number of iterations (per-node communications)
  that depends on~$\epsilon$
  (at least) as~$1/\epsilon^4$, and on $N$ as~$(\log N)^{O(1)}$ (poly-logarithmically).
     The algorithm in~\cite{Stateless} enjoys a stateless property (see~\cite{Stateless}
     for the definition of the property), while our algorithm
     is not stateless (e.g., it requires a global clock). References~\cite{Bienstock,FasterSimpler}
    develop algorithms with a better dependence on~$\epsilon$ than $1/\epsilon^4$. They are not
    concerned with developing fully distributed algorithms. The algorithm in~\cite{FasterSimpler}
     requires~$O(N \log N/\epsilon^2)$ iterations, while~\cite{Bienstock}
     takes~$O(\sqrt{N d_{\mathrm{max}} \log N}/\epsilon)$ iterations.
      In summary, among existing solvers~\cite{Luby}-\cite{MixedPackingCovering},~\cite{GlobalOptLocalInf,FasterSimpler},
      our algorithm matches the best dependence on~$\epsilon$, is fully distributed,
      can achieve arbitrary accuracy, is not stateless, and has in general worse dependence on~$N$.\footnote{Note that, for general networks, our algorithm's worst-case complexity is~$O(N^{3/2}/\epsilon)$.}


The remainder of the paper is organized as follows. The next paragraph
introduces notation. Section~\ref{section-allocations}
gives the system model with neighborhood-based data access and formulates
the distributed coded storage allocation problem as a FDS. Section~\ref{section-algorithm}
presents the {proximal center} distributed algorithm to solve FDS and states our results on its performance.
Section~\ref{section-proofs} gives the algorithm derivation and
proofs. Section~\ref{section-simulations} shows simulation examples. Finally, we conclude in Section~\ref{section-conclusion}.

Throughout, we use the following notation. We denote by: ${\mathbb R}^N$ the $N$-dimensional real space;
${\mathbb R}^N_+$ the set of $N$-dimensional vectors with non-negative entries;
 $a_i$ the $i$-th entry of vector $a$;
$A_{ij}$ the $(i,j)$-th entry of matrix~$A$; $0$ and $\mathbf{1}$
 the column vector with, respectively, zero and unit entries; $h_i$ the $i$-th canonical vector;%
 $\|\cdot\|=\|\cdot\|_2$ the Euclidean norm of a vector;
$\nabla g(y)$ the gradient at point $y$ of a differentiable function
 $g:\, {\mathbb R}^N \rightarrow \mathbb R$; $|S|$ cardinality of set $S$; and
 $\mathcal{I}_{\mathcal{E}}$ the indicator of event $\mathcal E$.
For two vectors $a, b \in {\mathbb R}^N$, the inequality
$a \leq b$ is understood component-wise. For a vector
 $a$, $[a]_+$ is a vector with the $i$-th entry equal to $\max\{0,a_i\}$ (Similarly,
 for a scalar $c$, $[c]_+ = \max\{0,c\}$.) Next,
 $\mathcal{P}_{[0,1]}(c)$ is the projection of scalar $c$ on the interval $[0,1]$, i.e.,
 $\mathcal{P}_{[0,1]}(c)=c$, for $c \in [0,1]$;
$\mathcal{P}_{[0,1]}(c)=0$, for $c < 0$; and
$\mathcal{P}_{[0,1]}(c)=1$, for $c>1.$
 Finally,
 for two positive sequences $\eta_n$ and $\chi_n$, $\eta_n = O(\chi_n)$ means that $\limsup_{n \rightarrow \infty}\frac{\eta_n}{\chi_n}<\infty$.

\vspace{-2mm}
\section{Problem model}
\label{section-allocations}
We consider coded distributed storage of a unit-size data object~$D$ over a network of
$N$ storage nodes. Each storage node $i$ stores a coded portion of $D$ of size $x_i$, $x_i \in [0,1]$.
For example, nodes can utilize random linear coding, where $D$ is divided into $M$
 disjoint parts; node $i$ stores $ x_i \times M $ random linear combinations
 of the parts of~$D$, e.g.,~\cite{Dimakis1} (ignoring
 the rounding of $ x_i \times M $ to closest integer). We assume that storage nodes
 constitute an arbitrary undirected network $\mathcal{G}=(\mathcal N, E)$,
 where $\mathcal N$ is the set of $N$ storage nodes, and
 $E$ is the set of communication links between them.
 Denote by $\Omega_i$ the one-hop closed neighborhood
 set of node $i$ (including~$i$),
 and by $d_i=|\Omega_i|$ its degree. Also,
 let $A$ be the $N \times N$ symmetric adjacency
 matrix associated with $\mathcal G$:
 $A_{ii}=1$, $\forall i$; and, for
 $i \neq j$, $A_{ij}=1$ if $\{i,j\} \in E$, and
 $A_{ij}=0$, else.

 A user accesses node~$i$ with probability $p_i>0$.
 Upon the user's request, node~$i$ contacts its neighbors
 $j \in \Omega_i \setminus \{i\}$,
 and they transmit their coded storage to $i$. Hence, afterwards,
 node $i$ has available the amount of storage
 equal to $\sum_{j \in \Omega_i} x_j$.
 If a MDS coding scheme is used, the recovery of object~$D$ is
 successful if $\sum_{j \in \Omega_i} x_j \geq 1$.
 Thus, the probability of successful recovery equals:
  $
  \sum_{i=1}^N p_i\,\mathcal{I}_{\{\sum_{j \in \Omega_i} x_j \geq 1\}}.
  $
 We aim at minimizing the total storage $\sum_{i=1}^N x_i$ such that
 the probability of recovery is one: $\sum_{i=1}^N p_i\,\mathcal{I}_{\{\sum_{j \in \Omega_i} x_j \geq 1\}}=1$.
  This translates into FDS, which, letting $x:=(x_1,...,x_N)^\top$, in compact form, can be written as:
\begin{equation}
\begin{array}[+]{ll}
\mbox{minimize} & {\mathbf 1}^\top x \\
\mbox{subject to} & A \, x \geq \mathbf 1 \\
                  & x \geq 0.
\end{array}
\label{LP}
\end{equation}
Clearly, \eqref{LP} has a non-empty constraint set (e.g., take $x=\mathbf 1$), and a solution exists.
Denote by $x^\star$ a solution to~\eqref{LP}.

In this paper, for simplicity, we focus on one-hop neighborhood data access model.
Our framework straightforwardly generalizes to $\ell$-hop neighborhood data access, $\ell >1$,
where a user attempts the recovery based on the $\ell$-hop
neighborhood of the queried node. Formally, in~\eqref{LP} we replace
$A$ with the adjacency matrix $A_{\ell}$ of graph $\mathcal{G}_{\ell}=(\mathcal N, E_{\ell})$,
where $E_{\ell}$ contains all pairs $\{i,j\}$ such
that there exists a path of length not greater than $\ell$ between them.

We are interested in developing a distributed, iterative algorithm,
where nodes exchange messages with their one-hop neighbors in the network,
so that the allocation $x$ produced by the algorithm
satisfies a $(1+\epsilon)$ approximation ratio: $\frac{1^\top x}{1^\top x^\star} \leq 1+ \epsilon$,
where $\epsilon>0$ is given beforehand.
  In this paper, we focus on how to determine the
(nearly optimal) \emph{amount} of coded storage $x_i$ at each node $i$.
 Once the amounts $x_i$'s are determined, in an actual implementation,
 nodes perform coding and actually \emph{store} the coded content; this is not considered here.\footnote{A simple way to
 achieve this, assuming each node $i$ knows $x_i$, is as follows.
 A data source passes the raw data object $D$ (partitioned into $M$ portions) to a randomly chosen node $i$.
 Then, node $i$ generates $x_i \times M$ random linear combinations, stores them, broadcasts
 the raw object $D$ to all its neighbors $j \in \Omega_i \setminus \{i\}$, and erases~$D$.
  Afterwards, each neighbor $j$ stores $x_j \times M$ random linear
 combinations, passes $D$ to all its neighbors unvisited so far, and erases~$D$.
 The process continues iteratively and terminates after all nodes have been visited; e.g.,
 it can stop after $N$ iterations.}

\vspace{-2mm}
\section{Distributed algorithm for coded storage allocation}
\label{section-algorithm}
%
%
%
%
\subsection{The algorithm}
\label{subsection-algorithm}
We apply the {proximal center} method~\cite{Necoara2} to solve the coded storage allocation problem~\eqref{LP}.
%
%
%
%
%
  The algorithm is based on the dual problem of a regularized version of~\eqref{LP},
 and on the Nesterov gradient algorithm.
 With respect to~\cite{Necoara2}, we choose the dual step-size differently; the step-size choice arises from the analysis here and in~\cite{Necoara}. Also, we modify the method to produce feasible primal updates at every iteration.
  (See Section~\ref{section-proofs} for the algorithm derivation.)

The algorithm is iterative, and all nodes operate in synchrony.
We denote the iterations by $k=0,1,2,...$ Each node $i$ maintains over iterations~$k$
its current (scalar) solution estimate $x_i^{(k)} $, where $x_i^{(k)} $ remains feasible to~\eqref{LP}, and the
auxiliary (scalar) variables: $\lambda_i^{(k)}$,
$\mu_i^{(k)}$, $\widehat{x}_i^{(k)}$, and $z_i^{(k)}$, and
$\overline{x}_j^{(k)}$, $\forall j \in \Omega_i.$ (Node~$i$ also has
$z_i^{(-1)}=0.$) 
  Let $1+\epsilon$ be the approximation ratio that
  nodes want to achieve, $\epsilon>0$.
   Our algorithm has parameters $\alpha,\delta>0$, set to:
   $\delta = \epsilon$,
   and $\alpha = \frac{\delta}{2(d_{\mathrm{max}}+1)^2}$.
   (See also Section~\ref{section-proofs}.)
    The algorithm is presented in Algorithm~1.
   We assume that all nodes know beforehand the quantities
    $d_{\mathrm{max}}$ and $\epsilon.$

       {\small{
\begin{algorithm}
\caption{Distributed algorithm for solving~\eqref{LP}}
\begin{algorithmic}[1]
{\small{
    \STATE (\textbf{Initialization}) Each node $i$ sets $k=0$,
    $\lambda_i^{(0)}=0$, $z_i^{(-1)}=0$,
    $\delta = \epsilon$,
   and $\alpha = \frac{\delta}{2(d_{\mathrm{max}}+1)^2}$.
        \STATE For $k=0,1,2,...$, perform steps $3$--$6$:
        \STATE If $k \geq 1$, each node $i$ transmits
        $\lambda_i^{(k)}$ to all $j \in {\Omega_i} \setminus \{i\}$,
        and receives $\lambda_j^{(k)}$, $\forall j \in {\Omega_i} \setminus \{i\}$ (if $k=0$ this step is skipped.)
        \STATE Each node $i$ computes:
        \begin{equation}
        \label{eqn-x-hat-i}
        \widehat{x}_i^{(k)} = \mathcal{P}_{[0,1]} \left( \frac{1}{\delta}(\sum_{j \in {\Omega_i}}\lambda_j^{(k)}-1)  \right)
        \end{equation}
        \STATE Each node $i$ transmits
        $\widehat{x}_i^{(k)}$ to all $j \in {\Omega_i} \setminus \{i\}$,
        and receives $\widehat{x}_j^{(k)}$, $\forall j \in {\Omega_i} \setminus \{i\}$.
        \STATE Each node $i$ computes:
        \begin{eqnarray}
        \label{eqn-z-i}
        z_i^{(k)} &=& z_i^{(k-1)} + \frac{k+1}{2} \left( 1-\sum_{j \in {\Omega_i}}\widehat{x}_j^{(k)} \right) \\
        \label{eqn-mu-i}
        \mu_i^{(k)} &=& \left[\, \lambda_i^{(k)} + \alpha \left( 1-\sum_{j \in {\Omega_i}}\widehat{x}_j^{(k)} \right) \, \right]_+ \\
        \label{eqn-lambda-i}
        \lambda_i^{(k+1)} &=& \frac{k+1}{k+3}  \,\mu_i^{(k)} + \frac{2}{k+3} \alpha \left[ z_i^{(k)} \right]_+ \\
        \label{eqn-x-bar-i}
        \overline{x}_j^{(k)} &=& \frac{k}{k+2} \overline{x}_j^{(k-1)} + \frac{2}{k+2} \widehat{x}_j^{(k)}, \, \forall j \in {\Omega_i} \\
        \label{eqn-x-i}
        x_i^{(k)} &=& \overline{x}_i^{(k)} + \left[ \,    1-\sum_{j \in {\Omega_i}}\overline{x}_j^{(k)}    \,  \right]_+
        \end{eqnarray}
}}
\end{algorithmic}
\vspace{-0mm}
\end{algorithm}}}
%

We can see that, with Algorithm~1, each node $i$: 1) performs two broadcast, scalar transmissions
to all neighbors, per $k$; 2) maintains $O(d_i)$ scalars over iterations~$k$ in its memory;
and 3) performs $O(d_i)$ floating point operations per~$k$. Here, $d_i$ is the degree of node~$i$.

When the generalized, $\ell$-neighborhood based data access is considered,
Algorithm~1 generalizes in a simple way: the structure remains the same, except that
the one-hop neighborhood $\Omega_i$ is replaced with the $\ell$-hop neighborhood in all
steps of the algorithm. Physically, this translates into requiring that nodes exchange messages
with all their $\ell$-hop neighbors during execution.

\subsection{Performance guarantees}
\label{subsection-guarantees}
We now present our results on the convergence and convergence rate of
Algorithm~1. We establish the following Theorem, proved in Section~\ref{section-proofs}.
\begin{theorem}
\label{theorem-1}
Consider Algorithm~1 with arbitrary $\delta>0$ and $\alpha = \frac{\delta}{2(d_{\mathrm{max}}+1)^2}$.
The iterates $x^{(k)}$
 are feasible to~\eqref{LP}, $\forall k=0,1,...$, and,
 for any solution $x^\star$ of~\eqref{LP}, $\forall k=0,1,...,$ there holds:
 {\small{
 \begin{eqnarray}
 \frac{  {\mathbf 1}^\top x^{(k)} - {\mathbf 1}^\top x^\star } {{\mathbf 1}^\top x^\star}
 &\leq&
 32 (d_{\mathrm{max}}+1)^3 (1+1/\delta)\, \frac{1}{(k+1)^2} \nonumber \\
 &+& \delta/2. \label{eqn-theorem-1}
 \end{eqnarray}}}
\end{theorem}
An immediate corollary of Theorem~\ref{theorem-1} is the following result. It can be easily proved
by setting both summands on the right hand side of \eqref{eqn-theorem-1} to~$\epsilon/2$.
\begin{corollary} Algorithm~1  with $\delta=\epsilon$ and $\alpha = \frac{\delta}{2(d_{\mathrm{max}}+1)^2}$
achieves the $(1+\epsilon)$-approximation ratio:
$\frac{  {\mathbf 1}^\top x^{(K_{\epsilon})} - {\mathbf 1}^\top x^\star } {{\mathbf 1}^\top x^\star} \leq \epsilon$ in $K_{\epsilon}=O\left( \frac{d_{\mathrm{max}}^{3/2}}{\epsilon} \right)$ iterations.
\end{corollary}

\vspace{-2mm}
\section{Algorithm derivation and analysis}
\label{section-proofs}
%
%
\subsection{Algorithm derivation}
\label{subsection-alg-derivation}
In this Subsection, we explain how Algorithm~1 is derived.
 A derivation for generic cost and prox functions can be found in~\cite{Necoara2}, but we
include the derivation specific to~\eqref{LP} for completeness.
We apply the Nesterov gradient method in the Lagrangian dual domain.
We first add the constraint $x \leq {\mathbf 1}$ in~\eqref{LP}.
Note that this can be done without changing the solution set. Next,
we introduce the
$l_2$-regularization, by adding the term $\frac{\delta}{2}\|x\|^2$ to the cost function.
 The regularization allows for certain nice properties of the
 Lagrangian dual, e.g., Lipschitz continuous gradient od the dual function.
 Hence, we consider the regularized problem:
\begin{equation}
\begin{array}[+]{ll}
\mbox{minimize} & {\mathbf 1}^\top x + \frac{\delta}{2} \|x\|^2\\
\mbox{subject to} & A \, x \geq \mathbf 1 \\
                  & 0 \leq x \leq \mathbf{1}.
\end{array}
\label{eqn-LP-regularized}
\end{equation}
By dualizing the constraint $\mathbf 1 - Ax \leq 0$,
we obtain the dual function:
 $\mathcal{D}: \,{\mathbb R}^N_+ \rightarrow {\mathbb R}^N$:
 \begin{equation}
 \label{eqn-dual-fcn}
 \mathcal{D}(\lambda)  = \min_{0 \leq x \leq \mathbf 1}
 \left\{  \mathbf{1}^\top x + \frac{\delta}{2}\|x\|^2 + \lambda^\top \left( \mathbf 1 - Ax\right)    \right\}.
 \end{equation}
The dual problem is then to maximize $\mathcal{D}(\lambda)$ over $\lambda \in {\mathbb R}^N_+$.

We apply the Nesterov gradient method on the dual function with zero initialization;
set $\lambda^{(0)}=0$, and, for $k=0,1,...,$ perform:
{\small{
\begin{eqnarray}
\label{eqn-alg-nesterov}
\mu^{(k)} &=&  \left[  \lambda^{(k)} + \alpha\, \nabla \mathcal{D}(\lambda^{(k)})\right]_+ \\
\lambda^{(k+1)}  &=&  \frac{k+1}{k+3}\mu^{(k)} + \frac{2}{k+3} \alpha \left[\sum_{s=0}^{k} \frac{(s+1)}{2}\nabla \mathcal{D}(\lambda^{(s)})\right]_+ \nonumber
\end{eqnarray}}}
Next, it can be shown that, for any $\lambda \in {\mathbb R}^N_+$, $\nabla \mathcal{D}(\lambda)
= 1-A \widehat{x}(\lambda)$, where:
\begin{eqnarray*}
\widehat{x}(\lambda) = \mathrm{arg\,min}_{0 \leq x \leq \mathbf 1}
 \left\{ \mathbf{1}^\top x + \frac{\delta}{2}\|x\|^2 + \lambda^\top \left( \mathbf 1 - Ax\right)   \right\}.
\end{eqnarray*}
It is easy to show that $\widehat{x}(\lambda) = (\widehat{x}_1(\lambda),...,\widehat{x}_N(\lambda))^\top$
 admits a closed form solution, with:
 \begin{eqnarray}
 \label{eqn-x-i-value}
 \widehat{x}_i(\lambda) = \mathcal{P}_{[0,1]} \left( \frac{1}{\delta}(\sum_{j \in {\Omega_i}}\lambda_j-1) \right),\,\,\forall i.
 \end{eqnarray}
 Next, note that $z^{(k)} = (z_1^{(k)},...,z_N^{(k)})^\top $ in~\eqref{eqn-z-i}
 is a recursive implementation of the sum $\sum_{s=0}^{k} \frac{(s+1)}{2}\nabla \mathcal{D}(\lambda^{(s)})$.
  Also, note that $\widehat{x}^{(k)}$ in~\eqref{eqn-x-hat-i} equals $\widehat{x}(\lambda^{(k)})$.
  Hence, we have derived the updates \eqref{eqn-x-hat-i}, \eqref{eqn-z-i},
  \eqref{eqn-mu-i}, and \eqref{eqn-lambda-i}, for $\widehat{x}^{(k)}$, $z^{(k)}$,
  $\mu^{(k)}$, and $\lambda^{(k)}$, respectively. It remains to explain the updates
  for $\overline{x}^{(k)}$ and $x(k)$. Regarding the quantity $\overline{x}^{(k)}$,
  we introduce it, as reference~\cite{Necoara} demonstrates that good optimality guarantees can be obtained for $\overline{x}^{(k)}$ (while such guarantees
  may not be achieved for $\widehat{x}^{(k)}$.) Finally, as
  $\overline{x}^{(k)}$ may be infeasible for~\eqref{LP} at certain iterations,
  we introduce $x^{(k)}$ in~\eqref{eqn-x-i} that are feasible by construction (See also the proof of this in Subsection~\ref{subsection-proofs}.)
  This completes the derivation.

\subsection{Auxiliary results and proof of Theorem~1}
\label{subsection-proofs}
%
%
%

We first derive certain  properties of~\eqref{LP} and~\eqref{eqn-LP-regularized}.
  Recall that the network maximal degree is $d_{\mathrm{max}}$ and let $d_{\mathrm{min}}$ be the minimal degree.
 Further, let $x^{\bullet}$ be the solution to~\eqref{eqn-LP-regularized}, and
 ${\lambda}^{\bullet}$ be an arbitrary solution to the dual of~\eqref{eqn-LP-regularized}
(maximization of $\mathcal{D}(\lambda)$ over $\lambda \in {\mathbb R}^N_+$).
 Also, recall that $x^\star$ is an arbitrary solution to~\eqref{LP}. We have the following Lemma.
\begin{lemma}
\label{lemma-properties-of-LPs}
There holds:
{\allowdisplaybreaks{
\begin{eqnarray}
\label{eqn-bounds-on-sum-x}
\frac{N}{d_{\mathrm{max}}+1} &\leq& {\mathbf 1}^\top x^\star \leq \frac{N}{d_{\mathrm{min}}+1} \\
\label{eqn-bounds-on-x-norm}
\|x^\bullet\| &\leq&  \|x^\star\| \leq \sqrt{N} \\
\label{eqn-bounds-on-lambda-norm}
\|\lambda^\bullet\| &\leq&  2 \delta \|x^\bullet\| + 2 \sqrt{N} \\
\label{eqn-bounds-4}
{\mathbf 1}^\top x^\bullet + \frac{\delta}{2} \|x^\bullet\|^2 &\leq& (1+\delta/2){\mathbf 1}^\top x^\star .
\end{eqnarray}}}
\end{lemma}
Note that~\eqref{eqn-bounds-4} means that $x^\bullet$ is a
$(1+\delta/2)$-approximate solution to~\eqref{LP}.
%
%
%
%

\begin{proof}
We first prove~\eqref{eqn-bounds-on-sum-x}. The lower bound~\eqref{eqn-bounds-on-sum-x} follows from Lemma~{4.1} in~\cite{Kuhn}.
The upper bound follows by noting that $x_i=1/(d_{\mathrm{min}}+1)$, $\forall i$, is feasible to~\eqref{LP}.
We now prove~\eqref{eqn-bounds-on-x-norm} and~\eqref{eqn-bounds-on-lambda-norm}. The right inequality in~\eqref{eqn-bounds-on-x-norm}
 holds because $x_i^\star \leq 1$, $\forall i$.
For the left inequality, note that
\begin{equation}
\label{eqn-aux}
\mathbf 1^\top x^\bullet + \frac{\delta}{2} \|x^\bullet\|^2
\leq
\mathbf 1^\top x^\star + \frac{\delta}{2} \|x^\star\|^2,
\end{equation}
 as $x^\bullet$ is the solution to~\eqref{eqn-LP-regularized}. 
 The left inequality in~\eqref{eqn-bounds-on-x-norm}
 now follows combining the latter with $\mathbf 1^\top x^\star \leq \mathbf 1^\top x^\bullet$,
 which holds as $x^\star$ is a solution to~\eqref{LP}. To prove~\eqref{eqn-bounds-on-lambda-norm},
 we use the Karush-Kuhn-Tucker conditions associated with~\eqref{eqn-LP-regularized}. In particular,
 they imply that:
$
(\mathbf 1 + \delta x^\bullet - A \lambda^\bullet )^\top y \geq 0,
$
 for all $0 \leq y \leq \mathbf 1$. Taking $y=h_i$, $i=1,...,N$, we get:
 $A \lambda^\bullet \leq \mathbf 1 + \delta x^\bullet$, from which the desired claim follows.
 Finally, we prove~\eqref{eqn-bounds-4}. We have:
 $
 \mathbf 1^\top x^\star + \frac{\delta}{2} \|x^\star\|^2 $ $=
 \sum_{i=1}^N x_i^\star \left( 1+\frac{\delta}{2}x_i^\star\right)$
 $\leq
 \left( 1+\frac{\delta}{2}\right)\mathbf 1^\top x^\star,
 $
where the inequality holds because $x_i^\star \in [0,1]$, $\forall i$.
Combining the latter with~\eqref{eqn-aux}, the result follows.
\end{proof}
Consider the dual function~$\mathcal{D}(\lambda)$ in~\eqref{eqn-dual-fcn}.
An important condition for~\eqref{eqn-alg-nesterov} (and hence for Algorithm~1) to work is that
the step size $\alpha$ be chosen in accordance with the Lipschitz constant of~$\nabla \mathcal{D}(\lambda)$.
   From Theorem~{3.1} in~\cite{Necoara2}, it follows that $\nabla \mathcal{D}(\lambda)$ is Lipschitz continuous with constant $L=\frac{(d_{\mathrm{max}}+1)^2}{\delta}$, i.e.,
  for all $\lambda^a,\lambda^b \in {\mathbb R}^N_+$:
  $
  \| \nabla \mathcal{D}(\lambda^a) - \nabla \mathcal{D}(\lambda^b) \| $
  $\leq L \, \|\lambda^a - \lambda^b\|.
  $
  We now borrow Theorems 2.9 and 2.10 in~\cite{Necoara} and adapt them to our setting.
   Denote by $e_i^{(k)}:= \left[ 1-\sum_{j \in {\Omega_i}} \overline{x}_j^{(k)} \right]_+$,
   and $e^{(k)}:=(e^{(k)}_1,...,e^{(k)}_N)^\top$.
  \begin{lemma}[\cite{Necoara}, Theorems 2.9 and 2.10]
  \label{lemma-Necoara}
  Consider Algorithm~1 with arbitrary $\delta>0$ and
  $\alpha = \frac{\delta}{2(d_{\mathrm{max}}+1)^2}$. Then, for $k=0,1,...$:
  \begin{eqnarray}
  \label{eqn-lemma-1}
  &\,& {\mathbf 1}^\top \overline{x}^{(k)}
  +
  \frac{\delta}{2}\|\overline{x}^{(k)} \|^2
  \leq
  {\mathbf 1}^\top x^\bullet
  +
  \frac{\delta}{2}\|x^\bullet \|^2 \\
  \label{eqn-lemma-2}
  &\,&
  \|e^{(k)}\| \leq \frac{16 L \|\lambda^\bullet\|}{(k+1)^2},
  \end{eqnarray}
  where $L=\frac{(d_{\mathrm{max}}+1)^2}{\delta}$.
  \end{lemma}
As a side comment,
fixing arbitrary~$k$ and assuming $\overline{x}^{(k)} \neq x^\bullet$,
 \eqref{eqn-lemma-1} means that $\overline{x}^{(k)}$ does not satisfy at least one of
the constraints $\sum_{j \in \Omega_i}\overline{x}^{(k)}_j \geq 1$, $i=1,...,N$
 (though the constraint violations all converge to zero as $k \rightarrow \infty$).
We are now ready to prove Theorem~\ref{theorem-1}.

\begin{proof}[Proof of Theorem~\ref{theorem-1}]
 Consider Algorithm~1, and fix arbitrary $k \geq 0.$
 Note that $x^{(k)} =
 \overline{x}^{(k)} + e^{(k)}$, and is by construction feasible to~\eqref{LP}.
 Indeed, for any $i$,
 $x_i^{(k)}$ is clearly non-negative. Also, as
 $e^{(k)}_j \geq 0$, $\forall j$, we have:
 $
 \sum_{j \in {\Omega_i}} x_j^{(k)}$ $\geq$
 $e^{(k)}_i +  \sum_{j \in {\Omega_i}} \overline{x}_j^{(k)} $
 $=$
 $\left[1-\sum_{j \in {\Omega_i}} \overline{x}_j^{(k)}  \right]_+$
 $ + \sum_{j \in {\Omega_i}} \overline{x}_j^{(k)} $
 $\geq$
 $\left(1-\sum_{j \in {\Omega_i}} \overline{x}_j^{(k)}  \right)$
 $+ \sum_{j \in {\Omega_i}} \overline{x}_j^{(k)} = 1.$
 %
 %
  Next, adding
 $\mathbf 1^\top e^{(k)}$ to both sides of~\eqref{eqn-lemma-1},
 and using $\frac{\delta}{2}\|\overline{x}^{(k)}\|^2 \geq 0$:
  $
 \mathbf 1^\top x^{(k)}$ $
 \leq
 \mathbf 1^\top x^\bullet + \mathbf 1^\top e^{(k)} $ $ + \frac{\delta}{2}\|x^\bullet\|^2.
  $
 Subtracting $\mathbf 1^\top x^\star$
  from both sides of this inequality,
   using
  $\mathbf 1^\top e^{(k)} \leq \sqrt{N}\|e^{(k)}\|$,
  and~\eqref{eqn-bounds-4}:
 $
  \mathbf 1^\top x^{(k)} $$- \mathbf 1^\top x^\star
 \leq$$
 \frac{\delta}{2}(\mathbf 1^\top x^\star )+ \sqrt{N} \|e^{(k)} \| .
  $
  %
%
The result now follows by applying \eqref{eqn-lemma-2},
using \eqref{eqn-bounds-on-x-norm} and \eqref{eqn-bounds-on-lambda-norm}, dividing the resulting
inequality by $\mathbf 1^\top x^\star$,
and using the left inequality in~\eqref{eqn-bounds-on-sum-x}.
\end{proof}
\vspace{-3mm}
\section{Simulation example}
\label{section-simulations}
This Section illustrates the performance of our
distributed algorithm for coded storage allocation and
compares it with the stateless distributed solver proposed in~\cite{Stateless}.
This is an efficient, easy to implement representative of existing \emph{distributed} methods~\cite{Luby,GlobalOptLocalInf,MixedPackingCovering}. We remark that the method in~\cite{Stateless}
is stateless, while ours is, like~\cite{Luby,GlobalOptLocalInf,MixedPackingCovering}, not stateless.

The simulation setup is as follows. The network of storage nodes is a geometric random graph.
Nodes are placed uniformly at random over unit square, and the node pairs within distance
$0.4$ are connected with edges. We consider two different values
of $N$: $N=100; 400$, and two different values of the target accuracy $\epsilon$:
$\epsilon = 1; 0.1$. Reference~\cite{Stateless} assumes
beforehand the knowledge of $N$ and $\epsilon$, not $d_{\mathrm{max}}$. Hence, for a
fair comparison (and some loss of our method), with our algorithm we replace $d_{\mathrm{max}}$ with $N$ (upper bound on $d_{\mathrm{max}}$)
 and set $\alpha = \frac{\delta}{2 N^2}$. Now, with both algorithms, their
 parameters depend on $N$ and $\epsilon$ only; we set them such that the guaranteed achievable
 accuracy with our method is $1+\epsilon/2$, and with~\cite{Stateless} is $1+6\epsilon$.
  This is in favor of~\cite{Stateless}, as faster convergence is achieved for lower required accuracy.
 We then look at how many per-node communications each algorithm requires to achieve
 the $1+6\epsilon$ accuracy. With both methods, we initialize
 the allocation to $x_i^{(0)}=1$, $\forall i$. With our method,
 we achieve this implicitly by initializing $\lambda_i^{(0)}=0$, $\forall i$.

 Theory predicts that, for a fixed $N$, our method performs better for
 sufficiently small $\epsilon$, while~\cite{Stateless} performs
 better for a sufficiently large $N$, for a fixed $\epsilon$. (Compare complexities
 $O(N^{3/2}/\epsilon)$ versus $O((\log N)^4/\epsilon^5)$.)
  Simulation examples show that, for moderate-size networks ($N = 100; 400$)
  our method is significantly faster already at coarse required accuracies ($\epsilon=0.1$).
  This is illustrated in Figure~1, which plots
  the relative error $(\mathbf 1^\top x^{(k)}-\mathbf 1^\top x^\star)/(\mathbf 1^\top x^\star)$ versus elapsed number of per-node scalar communications with the two methods. (Computational cost
  per communication with the two methods is comparable.)
   For $\epsilon=1$, $N=100$ (top left Figure), the algorithms are comparable; with the increase of $N$ to $400$,
   \cite{Stateless} becomes slightly better (bottom left Figure). However, for $\epsilon=0.1,$
   our method performs significantly better for both values of $N$ (See the two right Figures.)

 \begin{figure}[thpb]
      \centering
      \includegraphics[height=1.4 in,width=1.5 in]{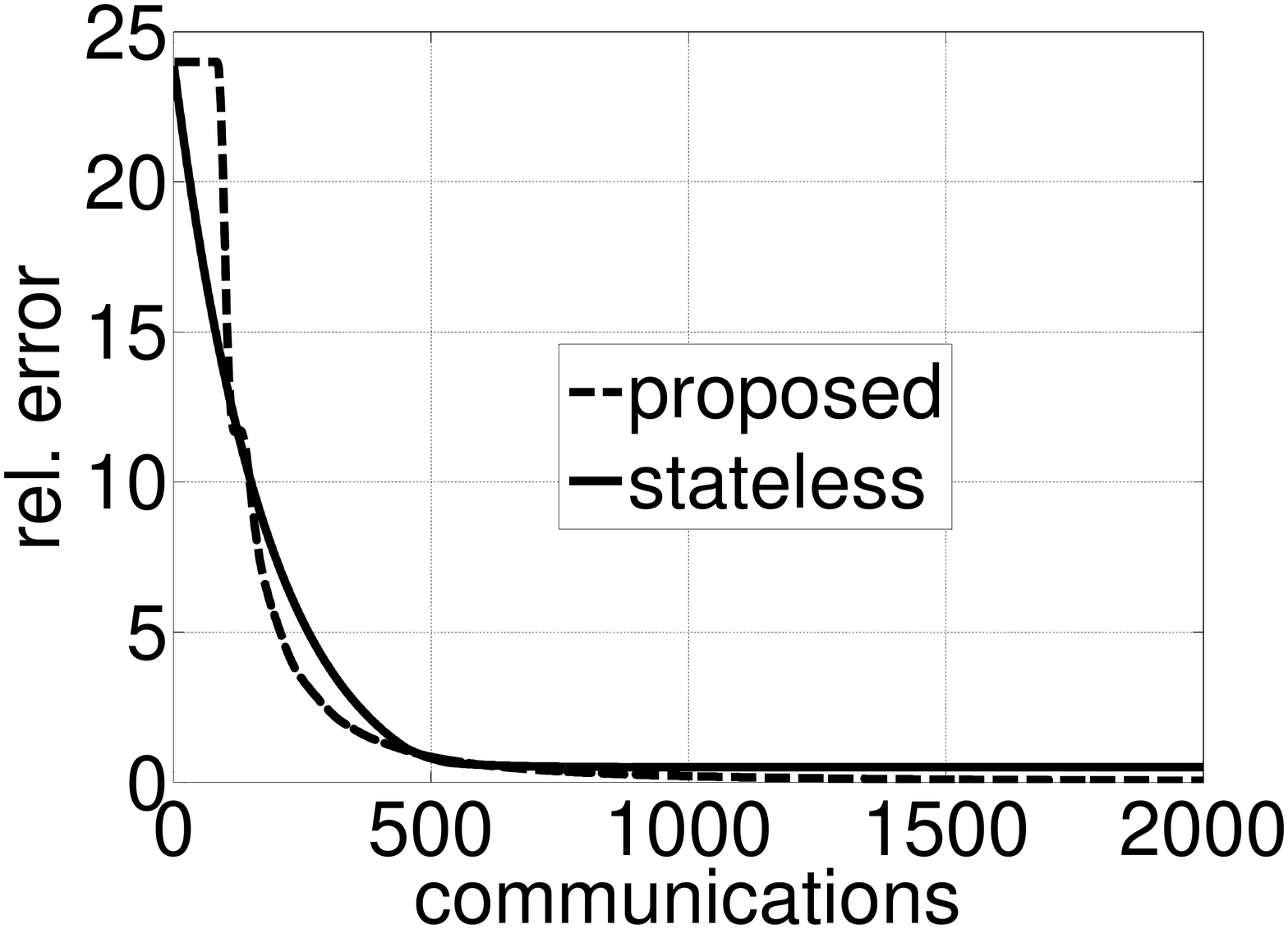}
       \includegraphics[height=1.4 in,width=1.5 in]{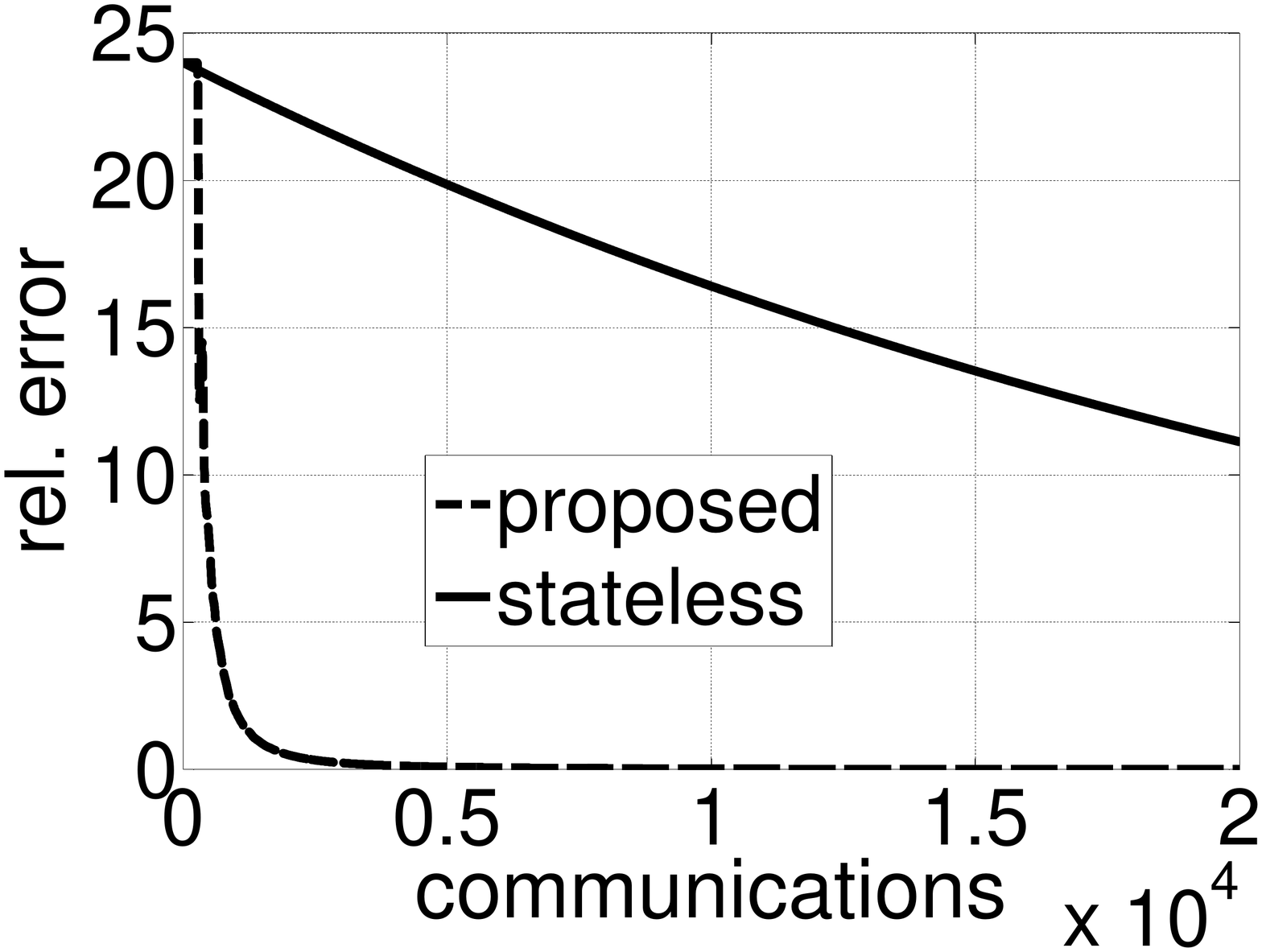}
       \includegraphics[height=1.4 in,width=1.5 in]{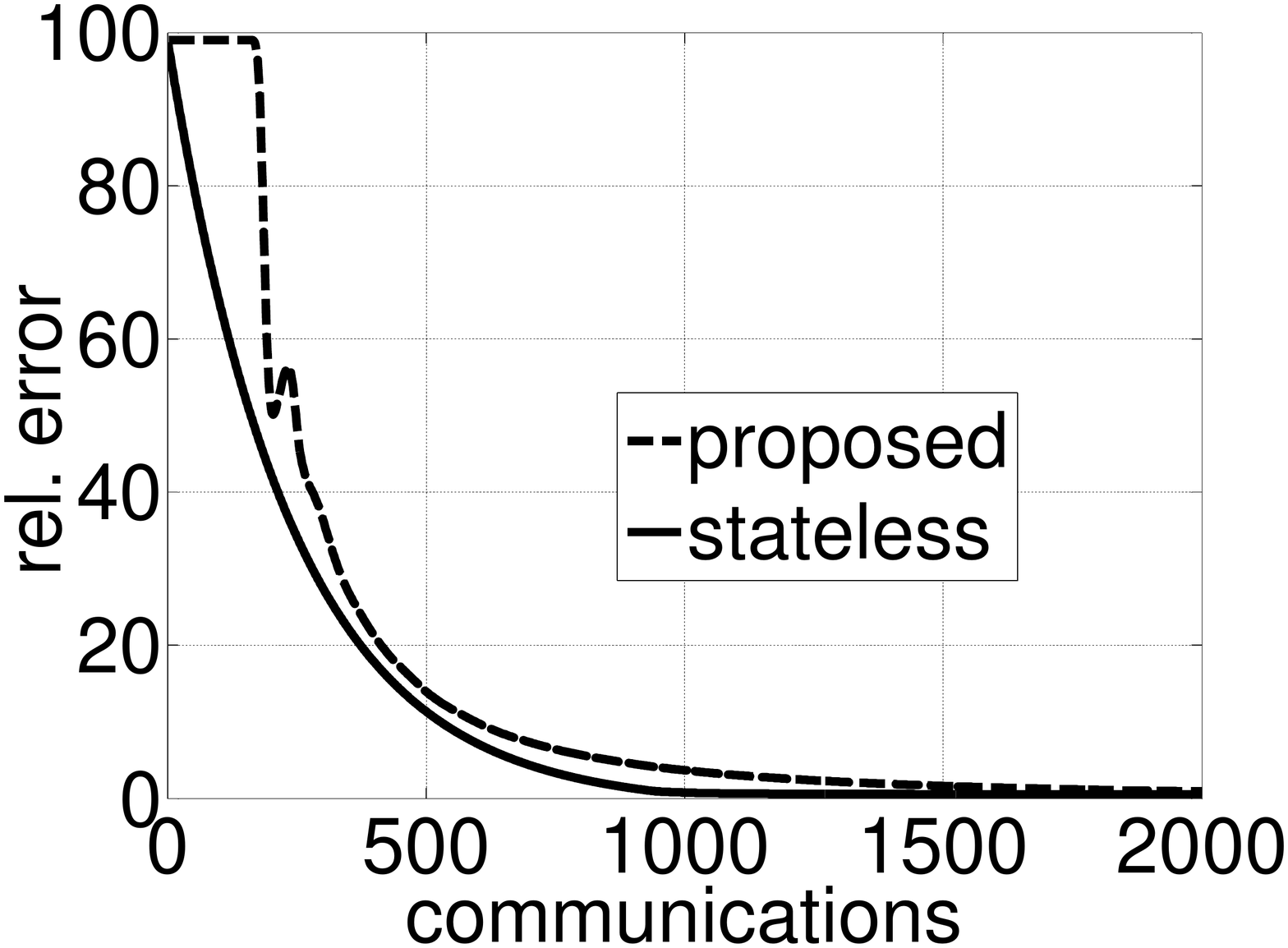}
       \includegraphics[height=1.5 in,width=1.6 in]{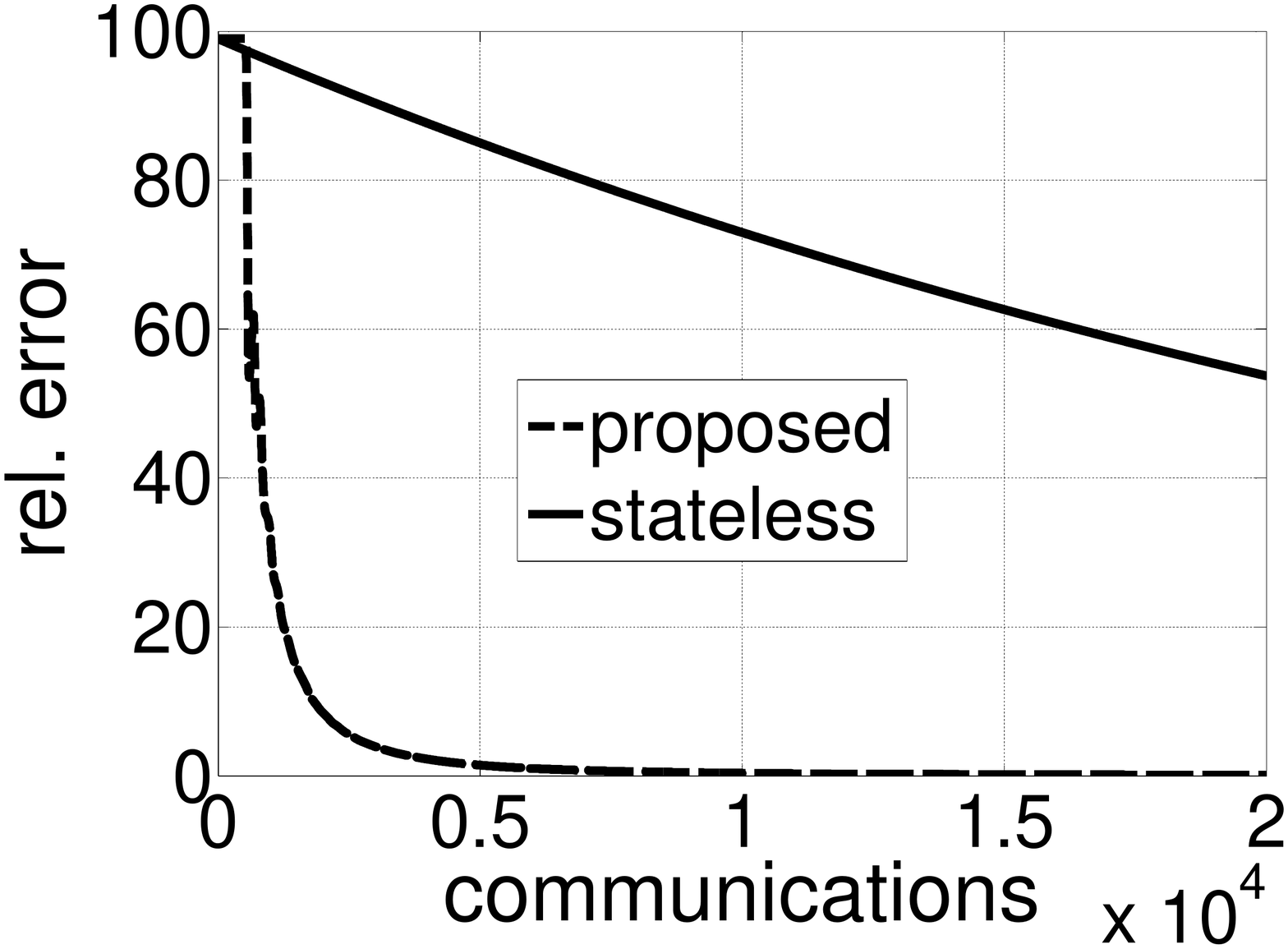}
      \caption{Relative error $(\mathbf 1^\top x^{(k)}-\mathbf 1^\top x^\star)/(\mathbf 1^\top x^\star)$ versus per-node communications along iterations for the
      proposed algorithm (dashed line) and the stateless algorithm in~\cite{Stateless} (solid line).
      Top left: $N=100, \epsilon=1$; top right: $N=100, \epsilon=0.1$;
       bottom left: $N=400,\epsilon=1$; bottom right: $N=400, \epsilon=0.1$.}
      \label{figure-1}
      \vspace{-7mm}
   \end{figure}

\vspace{-2mm}
\section{Conclusion}
\label{section-conclusion}
We introduced a new, neighborhood-based data access model for distributed coded storage allocation where
storage nodes are interconnected through a generic network. A user
accesses a randomly chosen storage node, and attempts
a recovery based on the storage available in the neighborhood set of the accessed node.
We formulate the problem of optimally allocating the coded storage
such that the overall storage is minimized while probability one recovery is guaranteed.
We show that the problem reduces to solving the fractional dominating set problem
over the storage node network. Next, we address the problem of designing an efficient fully distributed
algorithm to solve the coded storage allocation problem. While existing work did not focus
on Lagrangian dual methods, we apply and modify the dual {proximal center} method. We establish complexity of the method in terms of the desired accuracy and the underlying network
 and demonstrate its efficiency by simulations.

\vspace{-3.mm}
\bibliographystyle{IEEEtran}
\bibliography{IEEEabrv,bibliography_new}

\end{document}